\documentclass[final]{article}

\usepackage[margin=1in]{geometry}
\usepackage{amsmath}
\usepackage{amssymb}
\usepackage{amsthm}
\usepackage{graphicx}
\usepackage{epstopdf}

\usepackage{tabularx}
\usepackage{subfigure,booktabs,multirow,listings,paralist}

\newcommand\R{\mathcal R}

\newtheorem{theorem}{Theorem}
\newtheorem{lemma}[theorem]{Lemma}

\newtheorem{definition}[theorem]{Definition}
\newtheorem{observe}[theorem]{Observation}
\newtheorem{remark1}[theorem]{Remark}

\newenvironment{remark}{\begin{remark1} \rm}{\end{remark1}}

\usepackage{epsfig}

\def\R{\mathbb R}
\def\xb{{\bf x}}

\def\nub{{\boldsymbol \nu}}
\def\Bs{{\cal B}}
\def\Rs{{\cal R}}
\def\Cs{{\cal C}}
\def\Ds{{\cal D}}

\def\Hs{{\cal H}}
\def\tlambda{{\tilde\lambda}}

\def\tmu{{\tilde\mu}}
\def\talpha{{\tilde\alpha}}

\title{A fast multipole method for the evaluation of
elastostatic fields in a half-space with zero normal stress}

\author{Zydrunas Gimbutas\thanks{%
Information Technology Laboratory,
National Institute of Standards and Technology,
325 Broadway, Mail Stop 891.01,
Boulder, CO  80305-3328.
{{\em email}: {\sf {zydrunas.gimbutas@nist.gov}}}. 
Contributions by staff of NIST, an agency of the U.S. Government, 
are not subject to copyright within the United States.
}
\and Leslie Greengard\thanks{Courant Institute of Mathematical Sciences,
         New York University, 
         251 Mercer Street,
         New York, NY 10012-1110.
{{\em email}: {\sf {greengard@cims.nyu.edu}}}. }}

\begin{document}

\maketitle

\begin{abstract}
In this paper, we present a fast multipole method (FMM) for the 
half-space Green's function in 
a homogeneous elastic half-space subject to zero normal stress, 
for which an explicit solution was given by Mindlin (1936). 
The image structure of this Green's function is unbounded, so that standard
outgoing representations are not easily available.
We introduce two such representations here, one involving an expansion 
in plane waves and one involving a modified multipole expansion. Both
play a role in the FMM implementation.
\end{abstract}


\section{Introduction}


A classical problem in linear elasticity concerns the computation of
the displacement, stress and strain due to force and dislocation
sources with suitable boundary conditions imposed on a half-space.
The case of zero normal stress is of particular importance in
geophysical applications, for which the exact solution was derived by
Mindlin \cite{mindlin}.

We will concentrate here on the question of 
accelerating the evaluation of the field 
due to a collection of such force and dislocation vectors.
More precisely, we will describe two new analytic 
representations for the image structure in Mindlin's solution
that can be incorporated into a fast multipole method
(FMM). With $N$ sources and $M$ sensor/target locations, the FMM 
reduces the cost of evaluating the fields from $O(NM)$ to $O(N+M)$.
The FMM can also be used to accelerate
integral-equation based methods for elastostatic boundary value problems
on surfaces embedded in the half-space,
avoiding the ill-conditioning
associated with finite element and finite difference discretizations
of the underlying partial differential equations.

We will begin with a discussion of the mathematical foundations for 
the new scheme, followed
by a brief description of the full FMM implementation. 
For readers unfamiliar with fast multipole methods, we suggest the papers
\cite{fmmadap,fmmActa} to gain some familiarity, although the mathematical
treatment here is largely self-contained.

In section \ref{sec:numerical}, we present
numerical experiments both for collections of 
singular sources and for the evaluation of layer potentials using
the quadrature method of \cite{mindlinquad}.

\section{The Mindlin solution}

To fix notation, let us first consider the 
displacement at an observation
point $(x_1,x_2,x_3)$ due to a force 
vector ${\bf F} = (F_1,F_2,F_3)$ acting at the source point
$(\xi_1,\xi_2,\xi_3)$ in free-space.
The solution is given by the well-known Kelvin solution:
\[ {u_i} = {K_i^j} F_j \, ,  \]
where
\begin{equation}
 {K_i^j} = \frac{1}{8 \pi \mu}  \left\{  (2-\alpha)
\frac{\delta_{ij}}{r} + \alpha \frac{(x_i-\xi_i)(x_j-\xi_j)}{r^3} \right\} \, ,
\label{slpkernel}
\end{equation}
$r = \sqrt{ (x_1-\xi_1)^2 + (x_2-\xi_2)^2 + (x_3-\xi_3)^2}$, 
$\alpha = (\lambda+\mu)/(\lambda+ 2\mu)$,
and $\lambda,\mu$ are the Lam\'{e} coefficients. 
(In the preceding expressions, and throughout the paper,
we will generally make use of the standard summation convention.
On occasion we will write out the formulas explicitly when it 
makes the analysis clearer.)

Formulas for the strain ${\varepsilon_{ij}}$ and stress
${\sigma_{ij}}$ tensors can be obtained from partial derivatives of
the preceding formulas for displacement with respect to each component
$x_i$:
\begin{equation}
  {\varepsilon_{ij}} = \frac{1}{2} 
\left( \frac{\partial {u_i}}{\partial x_j} 
             + \frac{\partial {u_j}}{\partial x_i} \right) \, ,
\end{equation}
\begin{equation}
  {\sigma_{ij}} = \lambda \delta_{ij} \frac{\partial {u_n}}{\partial x_n} +
     \mu \left( \frac{\partial {u_i}}{\partial x_j} 
             + \frac{\partial {u_j}}{\partial x_i} \right) \, .
\end{equation}

A number of fast methods for the Kelvin solution
have been developed, based either on the FFT or the FMM 
\cite{fongdarve,frangi,duraiswami1,nishimura1,tornberg,wang2,wangwhite,zorin1}.

In a half-space,
the solution is more complicated, involving several image sources. 
We assume that the $x_3$-axis points up and that sources
$Q=(\xi_1,\xi_2,\xi_3)$ and targets $P=(x_1,x_2,x_3)$ are in the lower
half-space ($x_3,\xi_3<0$). 
With a slight modification of Okada's notation \cite{okada}, we let 
\[ R_1 = x_1-\xi_1 \ , \ R_2 = x_2-\xi_2 \ ,\ R_3 = -(x_3 + \xi_3), \]  
corresponding to the usual Cartesian components of the vector from the
image source $(\xi_1,\xi_2,-\xi_3)$ to the target, with the 
sign flipped in the $R_3$ component. Note that $R_3 \geq 0$.  
We denote the distance from the image to the target point by  
\[ R = \sqrt{R_1^2 + R_2^2 + R_3^2} \, . \]

Mindlin showed that the exact solution 
to the half-space problem with zero normal stress 
can be written in the form $u_i = W_i^j F_j$, where
\begin{align}
 W_i^j(P,Q) =\ &{K_i^j}(P,Q) + {A_i^j}(P,Q)  \nonumber \\
              & + {B_i^j}(P,Q) + x_3 {C_i^j}(P,Q) ,
              \label{mindlinfull}
\end{align}
with
\begin{align}
 {A_i^j} = & \frac{1}{8 \pi \mu}  \left\{  \alpha
\frac{\delta_{ij}}{R} + (2-\alpha) \frac{R_iR_j}{R^3} \right\} \, ,
\label{Aij} \\
{B_i^j} =  & \frac{1}{4 \pi \mu}  \left\{
\frac{1-\alpha}{\alpha}
\left[  \frac{\delta_{ij}}{R+R_3} +  
\frac{R_i \delta_{j3}  - R_j \delta_{i3}(1-\delta_{j3})}{R(R+R_3)}  \right. \right.  
\nonumber \\
&  \left. \left. -  \frac{R_i R_j}{R(R+R_3)^2} (1-\delta_{i3})(1-\delta_{j3}) \right] \right\} \, ,
\label{UijBC} \\
{C_i^j} = &  \frac{1}{4 \pi \mu}  (1- 2\delta_{i3}) \left\{
(2-\alpha) \frac{R_i \delta_{j3} - R_j \delta_{i3}}{R^3} + \right. \nonumber \\
& \left. \quad \alpha \xi_3
\left[ \frac{\delta_{ij}}{R^3} - \frac{3R_iR_j}{R^5} \right] \right\}.
\end{align}

\begin{definition}
We will refer to $W_i^j$ as the single-layer kernel in a half-space.
\end{definition}

The first contribution to  ${W_i^j}$ in formula (\ref{mindlinfull}) 
is the ``direct arrival" from the source in a uniform infinite medium,
given by the Kelvin formula (\ref{slpkernel}).
The second piece ${A_i^j}$ has the same form, but with the roles of
$\alpha$ and $(2-\alpha)$ reversed.  Since $R_3 = -x_3-\xi_3$, this is
the arrival at ``target" $(x_1,x_2,-x_3)$ from a source at
$(\xi_1,\xi_2,\xi_3)$ with modified Lam\'{e} coefficients.  Thus,
interactions governed by both the $K_i^j$ and $A_i^j$ contributions
can be computed using the ``free-space" single-layer kernel.
${B_i^j}$ and ${C_i^j}$ are quite different and their analysis is the
principal contribution of this paper.

\begin{remark}
A simple algebraic trick permits the computation of the $A_i^j$ 
contributions. Namely, we set
$\tlambda = \lambda+ 4\mu$ and 
$\tmu = -\mu$. It is easy to check that
\[ \talpha = (\tlambda+\tmu)/(\tlambda+ 2\tmu) = (2-\alpha), \quad
 (2-\talpha) = \alpha \, .
\]
Thus, 
\[ {A_i^j}(x_1,x_2,x_3) =  {K_i^j}[\tlambda,\tmu](x_1,x_2,-x_3) \, , \]
where $K_i^j[\tlambda,\tmu]$ denotes the Kelvin formula with the 
dependence on the Lam\'{e} coefficients made explicit.
\end{remark}

\begin{remark} \label{rmk_slp_signflip}
Note that the argument $x_3$ has been replaced by $-x_3$, so that
some care is required when evaluating terms such as
$\partial u_i/\partial x_l$ which appear in the stress and strain tensors.
\end{remark}

\begin{definition}
The double-layer kernel in a half-space is given by 
\begin{equation}
T_i^j = \left[ \lambda \delta_{jk} \frac{\partial W_i^n}{\partial{\xi_n}}
+ \mu \left( \frac{\partial W_i^j}{\partial{\xi_k}} + 
 \frac{\partial W_i^k}{\partial{\xi_j}} \right) \right] \nu_k.
\label{dlpkernel}
\end{equation}
This kernel describes the displacement field due to a dislocation vector 
${\bf D} = (D_1,D_2,D_3)$ across a surface $S$ with orientation vector 
${\bf \nub} = (\nu_1,\nu_2,\nu_3)$:
\begin{equation}
u_i = \int\int_S T_i^j D_j \, dS \, .
\label{dlpint}
\end{equation}
(Typically, the orientation vector is normal to the surface $S$.)
The dislocation vector ${\bf D}$ is sometimes called a {\em double-force 
vector}.
\end{definition}

To compute $T_i^j$, we note first that
\begin{align}
 \frac{\partial u_i^j}{\partial \xi_k}(x_1,x_2,x_3) = &
 \frac{\partial {K_i^j}}{\partial \xi_k}(x_1,x_2,x_3) + 
\frac{\partial {A_i^j}}{\partial \xi_k}(x_1,x_2,x_3)  \nonumber \\
              & + \frac{\partial {B_i^j}}{\partial \xi_k}(x_1,x_2,x_3) +
               x_3 \frac{\partial {C_i^j}}{\partial \xi_k}(x_1,x_2,x_3) ,
              \label{mindlinfulldlp}
\end{align}
where \cite{mindlin,okada,asme_rev}
\begin{align}
 \frac{\partial {K_i^j}}{\partial \xi_k} &= 
\frac{1}{8 \pi \mu}  \left\{  (2-\alpha)
\frac{(x_k-\xi_k)}{r^3} \delta_{ij} - \alpha \frac{(x_i-\xi_i) \delta_{jk}+
(x_j-\xi_j) 
\delta_{ik}}{r^3} \right. \nonumber \\
&\left. \qquad + 3 \alpha \frac{(x_i-\xi_i)(x_j-\xi_j)(x_k-\xi_k)}{r^5} 
\right\} \, , \\
\frac{\partial {A_i^j}}{\partial \xi_k} &= 
\frac{1}{8 \pi \mu}  \left\{  \alpha
\frac{R_k}{R^3} \delta_{ij} - (2- \alpha) 
\frac{R_i \delta_{jk} +R_j \delta_{ik}}{R^3} +
3 (2-\alpha) \frac{R_i R_j R_k}{R^5}  \right\} , 
\end{align}
\begin{align} 
 \frac{\partial {B_i^j}}{\partial \xi_k} &=  \frac{1}{4 \pi \mu}  \left\{
 - \frac{R_i \delta_{jk}  +  R_j \delta_{ik}  - R_k \delta_{ij}}{R^3} +
 \frac{3 R_i R_j R_k}{R^5}  \right.
 \nonumber \\
 &\qquad  +  \frac{1-\alpha}{\alpha}
\left[  \frac{ \delta_{3k}R + R+k}{R(R+R_3)^2} \delta_{ij} -  
\frac{ \delta{ik} \delta_{j3}  - \delta_{jk} \delta_{i3}(1-\delta_{j3})}{R(R+R_3)}   \right. \nonumber \\
&\qquad  + [R_i \delta_{j3}  - R_j \delta_{i3}(1-\delta_{j3})]
\frac{ \delta_{3k}R^2 + R_k (2R+R_3)}{R^3(R+R_3)^2}  
\label{Bijk} \\
&\qquad   +  
\left[ \frac{R_i \delta_{jk} + R_j \delta_{ik}}{R(R+R_3)^2}  - 
\right. \nonumber \\
&\qquad \qquad \left. \left. \left.
R_iR_j \frac{2 \delta_{3k}R^2 + R_k (3R+R_3)}{R^3(R+R_3)^3} \right]
(1-\delta_{i3})(1-\delta_{j3}) \right] \right\}, \nonumber 
\end{align}
\begin{align} 
 \frac{\partial {C_i^j}}{\partial \xi_k}  &=  \frac{1}{4 \pi \mu}  (1- 2\delta_{i3}) \left\{
(2-\alpha) \left[  \frac{\delta_{jk}\delta_{i3} - \delta_{ik} \delta_{j3}}{R^3} + \nonumber \right. \right. \\
&\qquad\qquad \left. \frac{3 R_k ( R_i\delta_{j3} - R_j \delta_{i3})}{R^5} \right] 
 + \alpha \left[ \frac{\delta_{ij}}{R^3} - \frac{3 R_i R_j}{R^5} \right] \delta_{3k} + \label{Cijk} \\
&\qquad\qquad \left. 3 \alpha \xi_3 \left[ \frac{R_i\delta_{jk} + R_j \delta_{ik} + R_k \delta_{ij}}{R^5} - 
\frac{5R_iR_jR_k}{R^7} \right] \right\}. \nonumber
\end{align}

As for the single-layer kernel,
\[ \frac{\partial {A_i^j}}{\partial \xi_k}(x_1,x_2,x_3)  =  
 \frac{\partial {K_i^j}}{\partial \xi_k}[\tlambda,\tmu](x_1,x_2,-x_3) \, . \]

We need to compute the contribution of 
$A_i^j$ to the double-layer kernel $T_i^j$ according to
(\ref{dlpkernel}):
\[ u_i = \left[ \lambda \delta_{jk} 
 \frac{\partial {A_i^n} }{\partial \xi_n} + \mu 
\left(  \frac{\partial {A_i^j}}{\partial \xi_k} +
 \frac{\partial {A_i^k}}{\partial \xi_j} \right) \right]
\nu_k D_j \, .
\]
Suppose that we invoke the {\em free-space} double-layer kernel 
with $\tlambda,\tmu$
and dislocation vector $-{\bf D}$, so that 
we actually compute
\begin{align*}
 u_i^{*} &= \left[ -\tlambda \delta_{jk} 
 \frac{\partial {A_i^n}}{\partial \xi_n} - \tmu 
\left(  \frac{\partial {A_i^j}}{\partial \xi_k} - 
\frac{\partial {A_i^k}}{\partial \xi_j} \right) \right]
\nu_k D_j  \\
 &= \left[ - \tlambda \delta_{jk} 
 \frac{\partial {A_i^n}}{\partial \xi_n} + \mu 
\left( \frac{\partial {A_i^j}}{\partial \xi_k} + 
\frac{\partial {A_i^k}}{\partial \xi_j} \right) \right]
\nu_k D_j \, .
\end{align*}
Fortunately, the difference is a simple harmonic function: 
\begin{align}
u_i - u_i^{*} & =  \left[ (\tlambda+ \lambda)  \delta_{jk} 
 \frac{\partial {A_i^n}}{\partial \xi_n}  \right]
\nu_k D_j  \nonumber \\
& = (\tlambda+ \lambda)   \frac{\partial {A_i^n}}{\partial \xi_n} 
({\bf \nub \cdot {D}}) \nonumber  \\
& = \frac{1}{8\pi\mu} (2\lambda + 4\mu) ({\bf \nub \cdot {D}})
\left\{  \alpha
\frac{R_i}{R^3}  - (2- \alpha) \frac{3R_i  +R_i }{R^3} +
3 (2-\alpha) \frac{R_i R^2}{R^5}  \right\}  \nonumber \\
& = \frac{1}{8\pi\mu} (2\lambda + 4\mu) ({\bf \nub \cdot {D}})
(2\alpha-2) \frac{R_i}{R^3}  \nonumber \\
& = \frac{1}{2\pi} ({\bf \nub \cdot {D}}) \frac{-R_i}{R^3} =  
\frac{1}{2\pi} ({\bf \nub \cdot {D}}) 
\frac{\partial}{\partial x_i} \frac{1}{R} \, . 
\label{diffcomp}
\end{align}

This difference can be computed using a single call to
the FMM for the 
Laplace equation, since the result is simply the gradient of the 
field due to a point source with strength
$({\bf \nub \cdot {D}})$.

\begin{remark}
For those keeping careful track of indices, note that, using
the Okada notation, it is indeed the gradient that is required.
We have moved $x_3$ to $-x_3$ in the free space call. Thus,
\[ \frac{\partial}{\partial x_3} \frac{1}{R} =  -(-x_3 - \xi_3)/R^3 = (x_3 + \xi_3)/R^3 = -R_3/R^3 \, ,
\] 
justifying the last equality in (\ref{diffcomp}).
\end{remark}

The difficulty in developing a fast algorithm for the Mindlin solution,
however, lies not in handling the free space kernel $K_i^j$ or the 
simple image $A_i^j$. Rather, it lies in the kernels $B_i^j$ and 
$C_i^j$.

\section{The $B$ image}\label{Bformulas}

Ignoring the scaling factor $\frac{1}{4\pi\mu} \, 
\frac{1-\alpha}{\alpha}$, the components of displacement induced by the
$B$ image can be written in the form:

\[ u_i^B = \frac{F_i}{R+R_3} +  
\frac{R_i F_3}{R(R+R_3)} -  
\frac{\delta_{i3}(F_1R_1 + F_2R_2)}{R(R+R_3)} -  
\frac{R_i(1- \delta_{i3})(F_1R_1+F_2R_2)}{R(R+R_3)^2}.
\]

Without entering into a detailed derivation, Mindlin's basic
observation was that the $B$ image could be derived from a consideration
of all second derivatives of a scalar potential. More precisely, we
have the following lemma.

\begin{lemma} \label{l3.1}
Let $\Bs(R_1,R_2,R_3)$ denote the scalar potential
given by $\Bs(R_1,R_2,R_3) = R_3 \log(R+R_3) - R$.
Then
\[
 \Bs_1 = \frac{-R_1}{R+R_3}, \qquad 
 \Bs_2 = \frac{-R_2}{R+R_3}, \qquad
 \Bs_3 = \log(R+R_3), 
\]
\begin{align*}
 \Bs_{11} &= \frac{-1}{R+R_3} + \frac{R_1^2}{R(R+R_3)^2}, \qquad
 \Bs_{12} = \frac{R_1 R_2}{R(R+R_3)^2}, \qquad 
 \Bs_{13} = \frac{R_1}{R(R+R_3)}, \\
  \Bs_{22} &= \frac{-1}{R+R_3} + \frac{R_2^2}{R(R+R_3)^2}, \qquad
 \Bs_{23} = \frac{R_2}{R(R+R_3)}, \qquad \Bs_{33} = \frac{1}{R}  \, ,
\end{align*}
where the subscript denotes differentiation with respect to the corresponding
variable $R_i$.
It follows that the contribution
to the displacement induced by the $B$ image in the 
single-layer kernel is given by 
\[ (u_1^B,u_2^B,-u_3^B) = \nabla_{\bf x} [ (F_1,F_2,F_3) \cdot
\nabla_{\overline \xi} \Bs ] \, ,\]
where $\nabla_{\bf x}$ denotes the gradient with respect to the 
target location and $\nabla_{\overline \xi}$ denotes the gradient with respect to the 
image source location at $(\xi_1,\xi_2,-\xi_3)$. 
\end{lemma}

\begin{remark} \label{slprmk}
Note that if $\Bs$ were the potential due to a simple charge source, 
then $(u_1^B,u_2^B,-u_3^B)$ would be the gradient of the potential
induced by a dipole with orientation and strength given by $(F_1,F_2,F_3)$.

\end{remark}

A straightforward but tedious calculation yields

\begin{lemma} \label{l3.2}
Let $\Bs(R_1,R_2,R_3) = R_3 \log(R+R_3) - R$. Then the contribution
to the displacement induced by the $B$ image in the 
double-layer kernel is given by 
\[ (u_1^B,u_2^B,-u_3^B) = \nabla_{\bf x} \Ds \, , \]
where 
\begin{align} \Ds &= 2\mu \left[ F_1\nu_1 \Bs_{\xi_1\xi_1} +
F_2\nu_2 \Bs_{\xi_2\xi_2} - F_3\nu_3 \Bs_{\xi_3\xi_3} +
(F_2\nu_1+F_1\nu_2) \Bs_{\xi_1\xi_2} \right]  \nonumber \\
&\qquad - 2\lambda ({\bf \nub \cdot {F}}) \Bs_{\xi_3\xi_3}.
\label{Bdlp}
\end{align}
\end{lemma}

\begin{remark} \label{dlprmk}
Note that the formula for $\Ds$ in Lemma \ref{l3.2} is, in essence,
a quadrupole field of a $\Bs$-type source with specific second 
derivative contributions defined 
in (\ref{Bdlp}).
\end{remark}

\subsection{Far field and local representations for the $B$ image}

It is easy to verify that $\Bs$ is a scalar harmonic function in the lower half-space.
It is also clear, however, that it cannot describe the field due to a bounded
collection of charges, since $\Bs$ is growing as 
$R_3 \rightarrow \infty$. In this section, we describe some new far field
representations that are somewhat involved, but permit much more efficient
computation.

\begin{figure}[htbp]{$ $}
\centerline{\psfig{figure=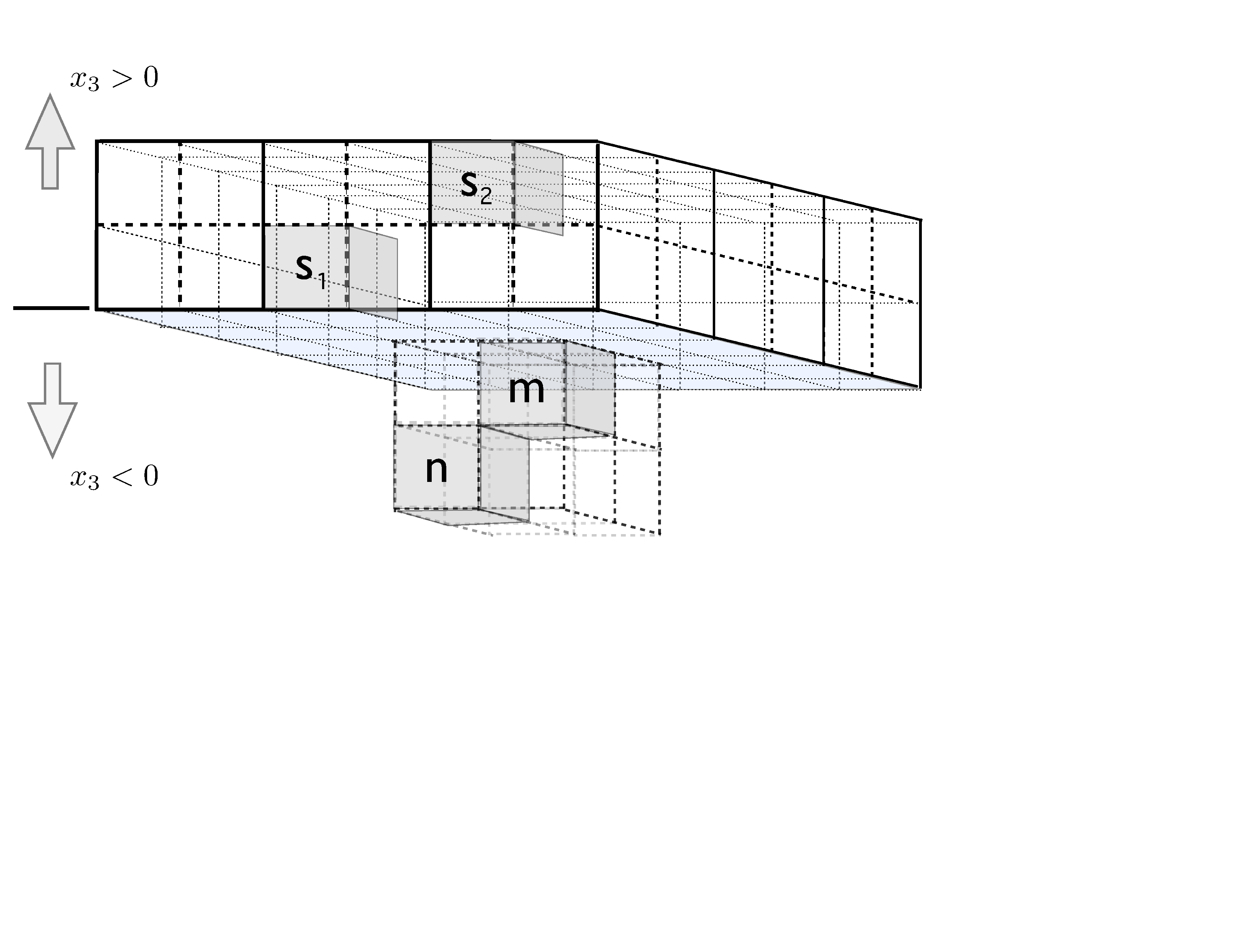,width=10cm}}

\vspace{-1in}

\caption{\sf The influence of the $\Bs$ (and $\Cs$) images 
in the boxes ${\bf s}_1$ and ${\bf s}_2$ needs to be computed at 
subsurface target locations in boxes 
${\bf n}$ and ${\bf m}$. In the FMM, this can be accomplished using 
a local expansion in the target boxes, an outgoing expansion in the 
source boxes, or both.
\label{downlist}}
\end{figure}

We begin by considering the $\Bs$-type sources contained in the 
boxes ${\bf s}_1$ 
and ${\bf s}_2$ in Fig. \ref{downlist}. They are separated from the 
target boxes ${\bf n}$ and ${\bf m}$ by at least one box length, so that
far field and/or local expansions should be rapidly convergent.
From above, the displacement in the lower half-space due to the image 
sources in, say, ${\bf s}_1$ is given by 
$(u_1^B,u_2^B,-u_3^B) = \nabla_{\bf x} \Phi_{\Bs}$, where the scalar $\Phi_{\Bs}$
is given by
\begin{equation}
\Phi_{\Bs} = \sum_{n=1}^N {\bf F}^{(n)} \cdot \nabla_{\overline \xi} 
\Bs(R^{(n)}_1,R^{(n)}_2,R^{(n)}_3),
\label{PhiBdef}
\end{equation}
where $N$ denotes the number of image sources in ${\bf s}_1$ and
$(R^{(n)}_1,R^{(n)}_2,R^{(n)}_3)$ denotes the vector from the $n$th
image source to the target point ${\bf x} = (x_1,x_2,x_3)$.

Within the box ${\bf n}$, however, the field $\Phi_{\Bs}$ is smooth and 
harmonic, and can be written in the form of a local expansion:

\begin{equation}
\Phi_{\Bs}(\xb) \approx \sum_{n=0}^p\sum_{m=-n}^n L_n^m \, 
	     Y_n^m(\theta,\phi) \, r^n ,
\label{local}
\end{equation}
with $(r,\theta,\phi)$ the spherical coordinates of $\xb$ with respect
to the box center of ${\bf n}$.
Here, 
$Y_n^m$ is the usual spherical harmonic of degree $n$ and order $m$
\begin{equation}
Y_n^m(\theta,\phi) =
\sqrt{\frac{2n+1}{4 \pi}}
\sqrt{\frac{(n-|m|)!}{(n+|m|)!}}\, P_n^{|m|}(\cos \theta)
		      e^{i m \phi},
\label{ynmpnm}
\end{equation}
where the associated Legendre functions $P_n^m$ are 
defined by the Rodrigues' formula
\[ P_n^m(x) = (-1)^m (1-x^2)^{m/2} \frac{d^m}{dx^m} P_n(x), \] 
and $P_n(x)$ is the Legendre polynomial of degree $n$.

The coefficients of the local expansion can be computed by projection
onto the spherical harmonic basis (integrating over the surface of
a sphere enclosing the box ${\bf n}$ and centered at the box center).
That is,
\begin{equation}
L_n^m= r^{-n} \, \int_0^{\pi} \int_0^{2\pi} Y_n^{-m}(\theta,\phi) 
\Phi_{\Bs}(r,\theta,\phi) \, d\phi \, d\theta.
\label{coeff}
\end{equation}
This can be carried out in $O(p^3)$ work, where $p$ is the order of the 
expansion in (\ref{local}) by using a tensor product grid with 
$2p$ Gauss-Legendre nodes in the $\theta$ variable and $2p$ equispaced
nodes in the $\phi$ variable.

In order to develop a more efficient fast algorithm, however, we would
like to have outgoing representations from the source box
${\bf s}_1$ that can make use of the full framework of the FMM
\cite{fmmadap,fmmActa}.
One such representation is based on the
plane wave formula (\cite{MF}, p. 1256) for the potential
at a target $(x_1,x_2,x_3)$
due to a simple charge source at $(\xi_1,\xi_2,-\xi_3)$:
\begin{eqnarray}
 \frac{1}{\sqrt{(x_1-\xi_1)^2+(x_2-\xi_2)^2+(x_3+\xi_3)^2}} 
\hspace{1.5in} \nonumber \\
 = \frac{1}{2\pi} \int_0^\infty e^{\sigma (x_3+\xi_3)} \int_0^{2 \pi}
e^{ i \sigma ((x_1-\xi_1) \cos \alpha + (x_2-\xi_2) \sin \alpha)} 
d\alpha \, d\sigma , 
\label{exprep} 
\end{eqnarray}
valid for $x_3,\xi_3 < 0$. 

The following theorem provides an expression for the displacement
induced by single and double-layer sources in terms of plane waves 
(that is, complex exponentials
of the components $(x_1,x_2,x_3)$).

\begin{theorem} \label{t3.1}
Let $(u_1^B,u_2^B,u_3^B)$ denote the displacement induced by a 
single-layer force vector $(F_1,F_2,F_3)$
located at the image source 
$(\xi_1,\xi_2,-\xi_3)$ that lies in a source box ${\bf s}$
centered at $(S_1,S_2,S_3)$.
Then
\[
u_i^B = \frac{1}{2\pi} 
\int_0^\infty {e^{\sigma (x_3-S_3)}} \int_0^{2 \pi}
e^{ i \sigma ((x_1-S_1) \cos \alpha + (x_2-S_2) \sin \alpha)} 
M_i(\alpha) W(\sigma,\alpha) \, d\alpha \, d\sigma,  
\]
where 
\[ W(\sigma,\alpha) = 
\left( -F_1 i \cos \alpha - F_2 i \sin \alpha + F_3 \right)
e^{\sigma (i(S_2 - \xi_2) + i(S_1 - \xi_1) + (S_3 + \xi_3))}
\]
and 
\begin{equation}
M_1(\alpha) = i \cos \alpha,\qquad
M_2(\alpha) = i \sin \alpha, \qquad M_3(\alpha) = -1.
\label{Sdef}
\end{equation}
\end{theorem}

Theorem \ref{t3.1} can be proven by Fourier analysis and contour
deformation, as in the derivation of the representation 
(\ref{exprep}) in \cite{MF}. 

\begin{remark} \label{r3.3}
Alternatively, we recall that $\Bs_{33} = \frac{1}{R}$. We
may write this relation in the form 
\begin{equation}
 \Bs = \partial_{x_3}^{-2} \left( \frac{1}{R} \right)  \, .
\label{dzminus2}
\end{equation}
It is straightforward to check that constants of integration can be 
ignored since they would only permit linear functions of 
$x_1$, $x_2$, and $x_3$ to appear in $\Bs$ and these are annihilated 
by the second derivative operators which arise in computing the 
displacement, according to Lemma \ref{l3.1}.
Note now that the operator $\partial_{x_3}^{-2}$ 
corresponds in (\ref{exprep}) to 
division by $\sigma^2$. This results in a divergent integral,
but using Lemma \ref{l3.1} again, 
the displacement clearly corresponds to 
multiplication
by a factor of either $\sigma \cos \alpha$, 
$\sigma \sin \alpha$, or $\sigma$ 
(the signatures of $\partial_{x_1}$, $\partial_{x_2}$, and
$\partial_{x_3}$, respectively).
This argument, of course, is not entirely rigorous,
but can be made so. 
\end{remark}

By superposition, we obtain a plane wave expansion for the 
field due to a set of sources, summarized in the following lemma.

\begin{lemma} \label{l3.3}
Let $(u_1^B,u_2^B,u_3^B)$ denote the displacement induced by a 
collection of single-layer force vectors 
$\{(F^n_1,F^n_2,F^n_3), n = 1,\dots,N \}$
at image source locations
\[ 
    \{ (\xi^n_1,\xi^n_2,-\xi^n_3), n = 1,\dots,N \},
\]
lying in a source box ${\bf s}$
centered at $(S_1,S_2,S_3)$.
Then the components of displacement are given by the plane
wave representation of Theorem \ref{t3.1}, with
\[ W(\sigma,\alpha) =  \sum_{n=1}^N
\left( -F^n_1 i \cos \alpha - F^n_2 i \sin \alpha + F^n_3 \right)
e^{\sigma (i(S_2 - \xi^n_2) + i(S_1 - \xi^n_1) + (S_3 + \xi^n_3))} .
\]
\end{lemma}

A plane wave expansion can be obtained for the double-layer kernel
as well. The proof is analogous.

\begin{theorem} \label{t3.2}
Let $(u_1^B,u_2^B,u_3^B)$ denote the displacement induced by a 
double-layer force vector ${\bf D} = (D_1,D_2,D_3)$ with orientation
vector ${\bf \nub} = (\nu_1,\nu_2,\nu_3)$,
located at the image source 
$(\xi_1,\xi_2,-\xi_3)$ that lies in a source box ${\bf s}$
centered at $(S_1,S_2,S_3)$.
Then,
\[
u_i^B = \frac{1}{2\pi} 
\int_0^\infty {e^{\sigma (x_3-S_3)}} \int_0^{2 \pi}
e^{ i \sigma ((x_1-S_1) \cos \alpha + (x_2-S_2) \sin \alpha)} 
M_i(\alpha) W(\sigma,\alpha) \, d\alpha \, d\sigma,
\]
where 
\begin{align*} W(\sigma,\alpha) &=  
 \sigma e^{\sigma (i(S_2 - \xi_2) + i(S_1 - \xi_1) + (S_3 + \xi_3))}
\, 
( 2\mu [ D_1 \nu_1 \cos^2 \alpha + D_2 \nu_2 \sin^2 \alpha 
  \\
 &\qquad - D_3\nu_3 + (D_2 \nu_1 + D_1\nu_2) \sin \alpha \cos \alpha 
] - 2 \lambda ({\bf \nub} \cdot {\bf D})  ),
\end{align*}
and the $M_i(\alpha)$ are defined in (\ref{Sdef}).
\end{theorem}

\begin{lemma} \label{l3.4}
Let $(u_1^B,u_2^B,u_3^B)$ denote the displacement induced by a 
collection of double-layer force vectors 
$\{{\bf D}^n = (D^n_1,D^n_2,D^n_3), n = 1,\dots,N \}$
at image source locations
$\{ (\xi^n_1,\xi^n_2,-\xi^n_3), n = 1,\dots,N \}$
with orientation vectors 
$\{ {\bf \nub}^n =(\nu^n_1,\nu^n_2,\nu^n_3), n = 1,\dots,N \}$
lying in a source box ${\bf s}$
centered at $(S_1,S_2,S_3)$.
Then the components of displacement are given by the plane
wave representation of Theorem \ref{t3.2}, with
\begin{align*} W(\sigma,\alpha) &=  \sum_{n=1}^N 
 \sigma e^{\sigma (i(S_2 - \xi^n_2) + i(S_1 - \xi^n_1) + (S_3 + \xi^n_3))}
\, 
( 2\mu [ D^n_1 \nu^n_1 \cos^2 \alpha + D^n_2 \nu^n_2 \sin^2 \alpha 
  \\
 &\qquad - D^n_3\nu^n_3 + (D^n_2 \nu^n_1 + D^n_1\nu^n_2)
 \sin \alpha \cos \alpha 
] - 2 \lambda ({\nub^n} \cdot {\bf D}^n)  ).
\end{align*}
\end{lemma}

Quadratures have been developed for these plane wave formulas 
in \cite{fmmActa,YARVIN}, valid so long as the source and target boxes are
separated in the $x_3$-direction by at least one intervening box length.
Referring to Fig. \ref{downlist}, ${\bf s}_1$ and ${\bf s}_2$
are well separated from ${\bf n}$ but only 
${\bf s}_2$ is well separated from ${\bf m}$. It is demonstrated in
\cite{fmmActa} that 3 digits of accuracy can be achieved with
about 100 plane waves,  
6 digits can be achieved with
about 560 plane waves,  
and 10 digits can be achieved with
about 1800 plane waves.

More concretely, suppose we wish to enforce a maximum error of $10^{-6}$.
Given a well-separated image source $Q=(\xi_1,\xi_2,-\xi_3)$ and target
$P=(x_1,x_2,x_3)$, we have
\begin{equation}
 \frac{1}{\|P-Q\|} \approx
 \sum_{k=1}^{18} \frac{w_k}{M(k)} \sum_{j=1}^{M(k)} 
e^{\sigma_k [(x_3+\xi_3) - i (x_1-\xi_1) \cos \alpha_j - 
i(x_2-\xi_2) \sin \alpha_j ]} , 
\label{disc18}
\end{equation}
where $\alpha_j = 2\pi j/M(k)$, and the 
weights $\{w_k\}$, nodes $\{ \sigma_k \}$
and values $\{ M(k) \}$ are given in 
Table \ref{tab2}.
(The total number of exponentials required is $558$.)
The weights and nodes $\{w_k,\sigma_k\}$ correspond to
a discretization of the 
outer integral in (\ref{exprep}). The
inner integral in (\ref{exprep}) is discretized using the 
trapezoidal rule with $M(k)$ nodes. The quadratures are designed
under the assumption that
$1 \leq |x_3+\xi_3| \leq 4$ and
$|x_1-\xi_1|,|x_2-\xi_2|  \leq 4$. 
This corresponds to their usage in the fast multipole method, 
where by convention, boxes at every level of the FMM 
hierarchy are rescaled to have unit size \cite{fmmadap,fmmActa}.

\begin{table}
\caption{Columns 1 and 2 contain the eighteen
weights and nodes for discretization of the outer integral in 
(\ref{exprep}) at six digit accuracy. Column 3 contains the 
number of discretization points needed in the inner integral,
denoted by $M(k)$ (From \cite{fmmActa}).
\label{tab2}}
\begin{center} 
\begin{tabular}{|r|r|r|}
\hline 
$Node\hspace{.5in}$ & $Weight\hspace{.5in}$ & $M(k)$ \\
\hline 
$0.05278852766117$  &
$0.13438265914335$ &
5 \\
$0.26949859838931$  &
$0.29457752727395$ &
8 \\
$0.63220353174689$ &
$0.42607819361148$ &
12 \\
$1.11307564277608$ &
$0.53189220776549$ &
16 \\
$1.68939496140213$ &
$0.61787306245538$ &
20 \\
$2.34376200469530$ &
$0.68863156078905$ &
25 \\
$3.06269982907806$ &
$0.74749099381426$ &
29 \\
$3.83562941265296$ &
$0.79699192718599$ &
34 \\
$4.65424734321562$ &
$0.83917454386997$ &
38 \\
$5.51209386593581$ &
$0.87570092283745$ &
43 \\
$6.40421268377278$ &
$0.90792943590067$ &
47 \\
$7.32688001906175$ &
$0.93698393742461$ &
51 \\
$8.27740099258238$ &
$0.96382546688788$ &
56 \\
$9.25397180602489$ &
$0.98932985769673$ &
59 \\
$10.25560272374640$ &
$1.01438284597917$ &
59 \\
$11.28208829787774$ &
$1.04003654374165$ &
51 \\
$12.33406790967692$ &
$1.06815489269567$ &
4 \\
$13.41492024017240$ &
$1.10907580975537$ &
1 \\
\hline 
\end{tabular}
\end{center}
\end{table}

The reason for seeking a plane wave representation for the displacement
due to a collection of sources is that {\em translation} of information
from a source box to a target box is a diagonal procedure.

\begin{lemma} {\bf (Diagonal translation)} [Adapted from \cite{fmmActa}]
\label{l3.5}
Let ${\bf s}$ be a box centered at $(S_1,S_2,S_3)$ containing
$N$ single and/or double-layer sources 
and let ${\bf n}$ be a well-separated
target box centered at $(N_1,N_2,N_3)$.
Suppose a plane wave expansion for the displacement takes the form
\begin{equation}
  u_i({ P}) = \sum_{k=1}^{s(\varepsilon)} \sum_{j=1}^{M(k)}
  M_i(\alpha_j) W(k,j) e^{\sigma_k (x_3-S_3)} 
 e^{i\sigma_k( (x_1-S_1) \cos \alpha_j + (x_2-S_2) \sin \alpha_j )}, 
\end{equation}
for $P = (x_1,x_2,x_3) \in {\bf n}$.
Then 
\begin{equation}
  u_i({ P}) = \sum_{k=1}^{s(\varepsilon)} \sum_{j=1}^{M(k)}
  V(k,j) e^{\sigma_k (x_3-N_3)} 
 e^{i\sigma_k( (x_1-N_1) \cos \alpha_j + (x_2-N_2) \sin \alpha_j )} \, , 
\end{equation}
where
\begin{equation}
 V(k,j) = W(k,j) \, e^{\sigma_k (N_3-S_3)}
 e^{i\sigma_k( (N_1-S_1) \cos \alpha_j + (N_2-S_2) \sin \alpha_j )}. 
\label{diagtrans}
\end{equation}
\end{lemma}

In the FMM, it is convenient to convert the plane wave expansion to a
local expansion in spherical harmonics of the form (\ref{local})
within a target box. For this, suppose we have the translated plane wave
expansion centered in the target box. Then, for the $i$th component
of displacement $u_i(P)$, we have
\cite{fmmActa}:
\begin{equation}
L_n^m = \frac{ (-i)^{|m|}}{\sqrt{(n-m)!(n+m)!}} \, 
 \sum_{k=1}^{s(\varepsilon)} (-\sigma_k)^n \, \sum_{j=1}^{M(k)}
  M_i(\alpha_j) W(k,j) e^{i m \alpha_j}.
\label{xtm}
\end{equation}

\subsection{An alternative representation}

While the representation in terms of plane waves above is, in a substantial
sense, optimal, it requires that the source and target
boxes be separated 
by a box length in the $x_3$-direction.
(In Fig.  \ref{downlist}, this condition fails for the interaction
of image source box ${\bf s}_1$ with target box ${\bf m}$.)
For these, we need either to make use of (\ref{PhiBdef}) and (\ref{local}) or
to find another far field representation.

One option would be to compute an equivalent density of $\Bs$-type sources
on the surface of a sphere enclosing the source box ${\bf s}_1$ as in
``kernel-independent" FMMs \cite{fongdarve,zorin1}. This is difficult to do
efficiently here since the kernel is not translation-invariant in $x_3$.

By combining the multipole expansion induced
by a collection of standard dipoles located at the 
$B$ image locations with dipole vectors $(F_1,F_2,F_3)$ with 
the formula (\ref{dzminus2}), it is easy to see that the following
lemma holds.

\begin{lemma} \label{l3.6}
Suppose we are given a collection of
single-layer force vectors 
$\{(F^i_1,F^i_2,F^i_3), i = 1,\dots,N \}$
at image source locations
$\{ Q_i = (\xi^i_1,\xi^i_2,-\xi^i_3), i = 1,\dots,N \}$
lying in a source box ${\bf s}$
centered at $(S_1,S_2,S_3)$.
Then $\Phi_{\Bs}$ is given by the far field representation
\begin{equation}
\Phi_{\Bs}(\xb) \approx \partial_{x_3}^{-2} 
\sum_{n=0}^p\sum_{m=-n}^n M_n^m \, 
	     Y_n^m(\theta,\phi)/ r^{n+1} ,
\label{mpole}
\end{equation}
where 
\begin{equation}
M_n^m = \sum_{i=1}^N 
(F^i_1,F^i_2,F^i_3) \cdot \nabla
(\rho_i^n\cdot Y_n^{-m}(\alpha_i,\beta_i))
   \label{mpcoeff}
\end{equation}
and $(\rho_i, \alpha_i, \beta_i)$ are the spherical coordinates of
$Q_i$ with respect to the center of $\bf s$.  A similar formula
for the multipole expansion induced by a collection of double-layer
sources can be obtained from Lemma \ref{l3.2}.
\end{lemma}

The difficulty with this representation is that 
$\partial_{x_3}^{-2}$ applied to 
a spherical harmonic $Y_n^m(\theta,\phi)/ r^{n+1}$
is a nonstandard special function  if $n-|m| \leq 1$.
To see why, we recall the following fact about spherical harmonics:

\begin{lemma} 
\label{dzminus2ynm}
Let $n-|m| \geq 2$. Then 
\begin{eqnarray*}
 \partial_{x_3}^{-2} \, Y_n^m(\theta,\phi)/ r^{n+1} = &  \\
\frac{\sqrt{2n+1}}{\sqrt{2n-3}} &
{\sqrt{\frac{1}{(n-m)(n-m-1)(n+m)(n+m-1)}}}
Y_{n-2}^m(\theta,\phi)/ r^{n-1}  \, .
\end{eqnarray*}
\end{lemma}

The proof is based on the well-known characterization of spherical harmonics 
as partial differential operators acting on $1/r$ (see, for example,
\cite{Thesis}).

To avoid this difficulty,
we will design two special rings of ``charge"
on the surface of a sphere enclosing the source box,
which will annihilate all multipole contributions of the form 
$Y^n_m(\theta,\phi)$
with $n-|m| \leq 1$.

\begin{lemma}
\label{ynmkill}
Let $\sigma_1$  denote a continuous distribution of $\Bs$-type 
sources on a ring $\Rs_1$ lying 
at the latitude corresponding to $\theta_1$ 
on the sphere of radius $R$.
Let $\sigma_2$  denote a continuous distribution of $\Bs$-type 
sources on a ring $\Rs_2$ of the same radius 
at the latitude corresponding to $\theta_2 = \pi - \theta_1$.
The multipole expansion induced by these rings of charge takes the form
(\ref{mpole}) with 
\begin{align*}
M_n^m = 
& \sqrt{\frac{2n+1}{4 \pi}}
\sqrt{\frac{(n-|m|)!}{(n+|m|)!}} R^{n+1} \sin \theta_1 \times \\
& \int_0^{2\pi} 
\left[ \sigma_1(\phi) P_n^{|m|}(\cos \theta_1) 
+ \sigma_2(\phi) P_n^{|m|}(-\cos \theta_1) 
\right] e^{-im\phi} \,
d\phi \, .
\label{coeffring}
\end{align*}
\end{lemma}

\begin{lemma}
\label{ynmkill2}
Suppose that $\Phi_{\Bs}$ is given by the far field representation
(\ref{mpole}) and let 
\[ \sigma_1(\phi) = \sum_{m=-n}^n \sigma_1^{(m)} e^{im \phi}, \]
\[ \sigma_2(\phi) = \sum_{m=-n}^n \sigma_2^{(m)} e^{im \phi}, \]
with $\sigma_1^{(m)}$ and $\sigma_2^{(m)}$
chosen to solve the linear system
\[  
\left( \begin{array}{cc} 1 & 1 \\ 1 & -1 \end{array} \right)
\left( \begin{array}{c} \sigma_1^{(m)} \\ \sigma_2^{(m)} \end{array} \right)
= 
\left( \begin{array}{c} 
M_m^m/( P_m^m(\theta_1) C_m \sin \theta_1 R^{m}) \\ 
M_{m+1}^m/( P_{m+1}^m(\theta_1) D_m \sin \theta_1 R^{m+1}) 
\end{array} \right) \, ,
\]
where 
\[ C_m = \frac{1}{4\pi[(2m)!]}, \quad 
D_m = \frac{1}{4\pi[(2m+1)!]}\, .
\]
Let 
\begin{equation}
\Psi(\xb) \approx \partial_{x_3}^{-2} 
\sum_{n=0}^p\sum_{m=-n}^n P_n^m \, 
	     Y_n^m(\theta,\phi)/ r^{n+1} 
\end{equation}
denote the far field expansion induced by the charge
distributions $\sigma_1^{(m)}$ and $\sigma_2^{(m)}$ lying on the 
rings defined in Lemma \ref{ynmkill}.
Then the multipole expansion of $\Phi_{\Bs} - \Psi$ takes the form
\begin{equation}
\Phi_{\Bs} - \Psi(\xb) \approx \partial_{x_3}^{-2} 
\sum_{n=0}^p\sum_{m=-n}^n (M_n^m - P_n^m) \, 
	     Y_n^m(\theta,\phi)/ r^{n+1} ,
\end{equation}
with 
\[ M_n^m-P_n^m = 0, \quad{\rm if}\ n-|m| \leq 1. \]
\end{lemma}

\begin{proof}
The result follows from Lemmas \ref{ynmkill}, the definition of 
$Y_n^m$ and some straightforward algebra.
\end{proof}

The point of this rather complicated computation is that we have
a new, efficient far field representation for $\Phi_{\Bs}$:
\begin{equation}
\Phi_{\Bs} = (\Phi_{\Bs} - \Psi) + \Psi \, ,
\label{comborep}
\end{equation}
where $\Psi$ is given by
\[ \Psi(\xb) = \int_0^{2\pi} \Bs(R_1,R_2,R_3) \sigma_1({\bf \eta^{(1)}}) 
d{\bf \eta^{(1)}}
+
\int_0^{2\pi} \Bs(R_1,R_2,R_3) \sigma_2({\bf \eta^{(2)}}) d{\bf \eta^{(2)}} 
\, ,
\]
and $\eta^{(1)},\eta^{(2)}$ are parametrizations of the rings $\Rs_1$
and $\Rs_2$ in Lemma \ref{ynmkill}. Here, $R_1 = x_1 - \eta_1$, $R_2 = x_2 -
\eta_2$, and $R_3 = -(x_3 + \eta_3)$, as usual in the Okada notation, and
Lemma \ref{dzminus2ynm} can be
used to construct a simple multipole expansion for the difference
$(\Phi_{\Bs}-\Psi)$.

\begin{lemma} \label{l3.10}
Suppose the function $\Phi_{\Bs}$
of Lemma \ref{l3.6} is expanded in the form 
(\ref{comborep}). Then,
a local expansion of $\Phi_{\Bs}$ can be computed using
$O(p^3)$ operations.
\end{lemma}

\begin{proof}
We first note that the smooth functions 
$\sigma_1({\bf \eta^{(1)}})$ and  
$\sigma_2({\bf \eta^{(2)}})$ can be sampled using $2p$ equispaced
points on the rings $\Rs_1$ and $\Rs_2$, since they have frequency content
bounded above by $p$. Computing the multipole expansions for 
$\Phi_{\Bs}$ and $\Psi$ requires only $O(p^2)$ work. Applying
$\partial_{x_3}^{-2}$ using Lemma \ref{dzminus2ynm} also requires $O(p^2)$
work. Mapping the multipole expansion to a local expansion in a target
box requires $O(p^3)$ work using the rotation-based scheme outlined in 
\cite{fmmActa}. Finally, the 
local expansion of $\Psi$ can be computed by evaluating $\Psi$ at 
$O(p^2)$ points on a sphere enclosing the target box and using 
the projection (\ref{coeff}). Both the evaluation and projection
step require $O(p^3)$ work.
\end{proof}

We will make use of the preceding results in section \ref{sec:numerical}.
Before that, however, we need to account for the field due to the ``$C$"
images.

\section{The $C$ image}\label{Cformulas}

Ignoring the $\frac{1}{4 \pi \mu}$ scaling, a little algebra shows that
the $C$ image contributions take the form:

\[ u_1^C = (2-\alpha) \frac{R_1F_3}{R^3} +  \alpha \xi_3 
\left[ \frac{F_1}{R^3} -  \frac{3R_1(F_1R_1+F_2R_2+F_3R_3)}{R^5} \right],
\]
\[ u_2^C = (2-\alpha) \frac{R_2F_3}{R^3} +  \alpha \xi_3 
\left[ \frac{F_2}{R^3} -  \frac{3R_2(F_1R_1+F_2R_2+F_3R_3)}{R^5} \right],
\]
\[ u_3^C = -(2-\alpha) \left( -
\frac{(R_1F_1+R_2F_2)}{R^3} \right) -
 \alpha \xi_3 
\left[ \frac{F_3}{R^3} -  \frac{3R_3(F_1R_1+F_2R_2+F_3R_3)}{R^5} \right].
\]
To simplify this, let $V = \frac{1}{R}$ and let 
subscripts on $V$ denote differentiation with respect to $R_i$.
Thus,
\[ V_i = \frac{-R_i}{R^3} \, , 
\qquad V_{ij} = \frac{3R_iR_j}{R^5} \ {\rm for}\ i\neq j, \qquad
 V_{ii} = \left( \frac{-1}{R^3} +\frac{3R_i^2}{R^5} \right). \]
A modest amount of algebra shows that, for the single-layer kernel, 
$(u_1^C,u_2^C,u_3^C) = \nabla_{\xb} \Phi_{\Cs} - (0,0,\Hs)$, where
\begin{equation} 
\Phi_{\Cs} = (2-\alpha)F_3 V - \alpha \xi_3 
[(F_1,F_2,-F_3) \cdot \nabla_{\overline \xi} V ],
\end{equation} 
\begin{equation}
 \Hs = -(2-\alpha) [(F_1,F_2,-F_3) \cdot \nabla_{\overline \xi} V ] \, ,
\end{equation} 
and $\nabla_{\overline \xi}$ denotes the gradient with respect to the 
image source location at $(\xi_1,\xi_2,-\xi_3)$.
For the double-layer kernel, the $C$ image contribution takes the form
\[ 
   (u_1^C,u_2^C,u_3^C) = \nabla_{\xb} \Phi_{\Cs} - (0,0,\Hs),
\]
where
\begin{align} \Phi_{\Cs} = &-2\alpha\xi_3\mu \left[ F_1\nu_1 V_{\xi_1\xi_1} +
F_2\nu_2 V_{\xi_2\xi_2} + F_3\nu_3 V_{\xi_3\xi_3} +
(F_2\nu_1+F_1\nu_2) V_{\xi_1\xi_2} \right]  \nonumber \\
&+ 2\alpha\xi_3\mu \left[ (F_1\nu_3 + F_3 \nu_1) V_{\xi_1\xi_3} +
(F_2\nu_3+F_3\nu_2) V_{\xi_2\xi_3} \right] 
 +  {\bf G} \cdot \nabla_{\xi} V ,
\end{align}
with
\[ {\bf G} = (2\alpha-2) [ \mu(F_1\nu_3+\nu_3F_1, 
F_2\nu_3+\nu_3F_2,-2F_3\nu_3) - (0,0,
\lambda (\nub \cdot \bf{F})) \, 
\]
and 
\begin{align} 
\Hs &= -2(2-\alpha)\mu \left[ F_1\nu_1 V_{\xi_1\xi_1} +
F_2\nu_2 V_{\xi_2\xi_2} + F_3\nu_3 V_{\xi_3\xi_3} +
(F_2\nu_1+F_1\nu_2) V_{\xi_1\xi_2} \right.  \nonumber \\
& \left. -(F_1\nu_3 + F_3 \nu_1) V_{\xi_1\xi_3} -
(F_2\nu_3+F_3\nu_2) V_{\xi_2\xi_3} \right]  \, .
\end{align}

Since the functions that appear here are all derivatives of $V = \frac{1}{R}$,
they fall within the framework of the FMM for the Laplace equation, and 
we omit detailed formulas about the formation and manipulation
of multipole expansions.
They can be found in \cite{fmmadap,fmmActa}.

\section{Informal description of the FMM}

The most straightforward implementation of a fast multipole method 
for the Mindlin solution is to call an evaluation routine three times:
once of the free-space (Kelvin-type) interactions, once for the $A$ images,
and once for the $B$ and $C$ images. By placing sources or targets
at the image locations for the $A$, $B$ and $C$ interactions, we can
carry out each calculation as if it were in free space.

In the FMM, one begins by defining the computational domain to be
the smallest cube in $\R^3$ containing all sources and targets
\cite{fmmadap}.
This is defined to be refinement level 0.
The domain is then subdivided into smaller and smaller boxes. More precisely,
refinement level $l+1$ is obtained from
level $l$ by subdividing each box $b$ at level $l$ into eight cubic boxes of
equal size. These small boxes are said to be {\em children} of $b$.
and $b$ is referred to as their {\em parent}. 
This recursive process is halted when a box contains fewer
than $s$ sources and/or targets, where $s$ is a free parameter.
Such boxes are referred to as leaf nodes and are {\em childless}.
If the box under consideration contains no sources or targets, 
it is deleted from the data structure. 

\begin{definition}
Two boxes at the same refinement level 
are said to be {\em colleagues} if they 
share a boundary point. A box is considered to be a colleague of itself.
The set of colleagues of a box $b$ will be denoted by $Coll(b)$.
\end{definition}

\begin{definition}
Two boxes are said to be {\em well separated} if they are at the same
refinement level and are not colleagues.
\end{definition}

\begin{definition}
\label{intlist}
With each box $b$ is associated an {\em interaction list}, consisting
of the children of the colleagues of $b$'s parent which are
well separated from box $b$.
\end{definition}

Note that a box can have up to 
27 colleagues and that its interaction list contains
up to 189 boxes.

\begin{definition}
{\em List 1} of a childless box $b$, denoted by $L_1(b)$,
 is defined to be the set
consisting of $b$ and all childless boxes adjacent to $b$.
If $b$ is a parent box, its List 1 is empty. 
\end{definition}

\begin{definition}
{\em List 2} of a box $b$, denoted by $L_2(b)$,
 is the set consisting of
all children of the colleagues of $b$'s parent  that are
well separated from $b$. 
\end{definition}

\begin{definition}
{\em List 3} of a childless box $b$, denoted by $L_3(b)$,
is the set consisting of
all descendents of $b$'s colleagues that are 
not adjacent to $b$, but whose parent boxes are
adjacent to $b$.
If $b$ is a parent box, its list 3 is empty.
\end{definition}

Any box $c$ in $L_3(b)$ is smaller than $b$ and is
separated from $b$ by a distance not less than the side of $c$,
and not greater than the side of $b$.

\begin{definition}
{\em List 4} of a box $b$, denoted by $L_4(b)$,
consists of boxes $c$ such that
$b \in L_3({c})$; in other words, $c \in L_4(b)$
if and only if $b \in L_3(c)$.
\end{definition}

\vspace{.1in}
{\sf
\vspace{.1in}
\hrule
\vspace{.1in}
\begin{center} Adaptive FMM for the $B$ and $C$ images  \end{center}
\vspace{.1in}
\hrule
\vspace{.1in}

\noindent
{\bf Initialization}

\noindent
Choose precision $\varepsilon$
and the order of the multipole expansions  $p$.
Choose the maximum number $s$ of charges allowed in a childless box.
Define $B_0$ to be the smallest cube containing all sources/targets
(the computational domain).

\vspace{.1in}

\noindent
{\bf Build Tree Structure} \\
\noindent
{Step 1} \\
The greatest refinement level is denoted by 
$L_{max}$ and the total
number of boxes created is denoted by $N_B$. 
Create the four lists for each box.

\vspace{.1in}

\noindent
{\bf Upward Pass} \\
(During the upward pass, $p$th-order multipole expansions are formed for
 each box $b$ containing image sources.)

\vspace{.1in}

\noindent
{Step 2a} \\
For each childless box that contains image sources,
use Lemma \ref{l3.6} to form the $p$th-order multipole expansion  
for $\Phi_{\Bs}$ and standard multipole formulas to form
expansions for $\Phi_{\Cs}$,$\Hs$.

\vspace{.1in}

\noindent
{Step 2b} \\
Beginning with the leaf nodes, carry out an upward recursion to shift
each multipole expansion to the parent's center.

\vspace{.1in}

\noindent
{\bf Downward Pass} \\
During the downward pass, a $p$th-order local expansion is generated for
each box $b$ about its center, representing the potential 
in $b$ due to all charges outside $(L_1(b) \cup L_3(b))$.

\vspace{0.1in}

\noindent
{Step 3} \\
For each box $b$, add to its local expansion the
 contribution due to $\Bs$ and $\Cs$-type sources in $L_4(b)$.
This can be done using 
the projection formula (\ref{coeff}) for $\Phi_{\Bs}$ and
from standard formulas \cite{fmmadap} for $\Phi_{\Cs}$ and $\Hs$.

\vspace{0.1in}

\noindent
{Step 4} \\
For each box $b$ containing sources and each box $c \in L_2(b)$ containing
targets, transmit far field information from $b$ to $c$. 

\vspace{0.1in}
\noindent
If the boxes are separated in the $x_3$-direction,
this is accomplished by 
converting the multipole expansions to plane wave expansions using
Theorem 3.3 of \cite{fmmadap}. These plane wave expansions are 
translated in diagonal form using Lemma \ref{l3.5}. 
For the $\Phi_{\Bs}$ expansion, the operator $\partial_{x_3}^{-2}$
can be applied as discussed in Remark \ref{r3.3}. 

\vspace{0.1in}
\noindent
If the boxes are not separated in the $x_3$-direction,
then the $\Phi_{\Cs}$ and $\Hs$ expansions can still be translated using 
any ``multipole-to-local" translation operator.
For the $\Phi_{\Bs}$ expansion, use Lemma \ref{l3.10}.

\vspace{0.1in}
\noindent
Once all plane wave expansions have been received by a given box $c$
for $\Phi_{\Bs}$, $\Phi_{\Cs}$ and $\Hs$,
use Theorem 3.4 of \cite{fmmadap} to 
convert each of the net plane wave expansions
into a local expansion and add to the corresponding local 
expansions associated with box $c$.

\vspace{0.1in}

\noindent
{Step 5} \\
For each parent box $b$, 
shift the center of its local expansions to its children.

\vspace{0.1in}

\noindent
{\bf Evaluation of displacement, stress and and strain}

\vspace{0.1in}

\noindent
{Step 6} \\
For each target in each childless box $b$
compute contribution to displacement, stress and strain
from local expansions in $b$.

\vspace{0.1in}

\noindent
{Step 7} \\
For each childless box $b$, calculate the contribution to the 
displacement, stress and strain directly from all 
image sources in $L_1(b)$.

\vspace{0.1in}

\noindent
{Step 8} \\
For each childless box $b$, and for each box $c \in L_3(b)$,
calculate the displacement, stress and strain at each target in $b$
from the multipole expansions for $\Phi_{\Bs}, \Phi_{\Cs}$ and 
$\Hs$. For $\Phi_{\Bs}$, this is done using Lemma \ref{l3.10}.
}
\vspace{.1in}
\hrule

\begin{remark}
We have not included detailed formulas for stress and strain here
since they are quite lengthy and not very informative. They
involve derivatives of the displacement vector, and 
the FMM provides a natural framework for this calculation.
One simply differentiates the 
local spherical harmonic expansions in each target box to obtain
the far field contributions. 
\end{remark}

\section{Numerical experiments}\label{sec:numerical}

In this section, we present timing results for the elastostatic FMM in
a half-space. All calculations were carried out using double-precision
arithmetic on a 3.1 GHz Xeon workstation with 128GB of RAM. For
comparison, we also present timings for the underlying harmonic FMM
and for the free space elastostatic FMM.  In all tables below, $N$
denotes the number of sources, $Prec$ denotes the precision parameter
(the number of digits requested from the FMM), $T^{harm}_{FMM}$
denotes the time required by the harmonic FMM for dipole sources,
$T^{elast}_{FMM}$ denotes the time required by the free space
elastostatic FMM for double-layer sources, $T^{Mindlin}_{FMM}$ denotes
the time required by the half-space elastostatic FMM for double-layer
sources, and $T^{harm}_{dir}$, $T^{elast}_{dir}$, $T^{Mindlin}_{dir}$
denote the times required by direct summation methods. The direct
timings are estimated from the actual timings using $N/100$ sources.
To make timings comparable, we compute the potentials together with
their first and second derivatives in the harmonic FMM, and the
displacements and strains in both the free and half-space elastostatic
FMMs.

To test the performance of the scheme, we carried out experiments
with sources distributed randomly on the surface of a cylinder with
unit radius and unit height. We denote the relative $L_2$ errors at
$N/100$ evaluation locations by $E_{harm}$ for the computed potentials
in the harmonic FMM, and $E_{elast}$, $E_{Mindlin}$ for the computed
displacements in the free and half-space elastostatic FMMs.

As expected, the FMM scales approximately linearly and the work
required for the free space elastostatic FMM is approximately 4 times
greater than for a corresponding harmonic FMM. The timing analysis for
the half-space elastostatic FMM is more complicated due to additional
FMM calls used to process the $A$, $B$, and $C$ images. Since these images are
well-separated from the evaluation locations, the local interaction
work is typically smaller, yielding slightly better timings for
this part of the calculation. In our implementation, the total work
required for the half-space elastostatic FMM is approximately 7 to 8 times
greater than for the corresponding harmonic FMM.

\begin{table}[h]
\caption{Timing results for harmonic dipoles and elastostatic double-layer 
sources in free space.
\label{tab3}}
\begin{center} 
\begin{tabular}{r|r|r|r|r|r|r|r}

\hline 

$N$ & $Prec$ & $T^{harm}_{FMM}$ & $T^{harm}_{dir}$ & $E_{harm}$ & $T^{elast}_{FMM}$ & $T^{elast}_{dir}$ & $E_{elast}$ \\

\hline 

 5000 &  2 & 0.140 & 1.248 & 1.244e-06 & 0.608 & 6.448 & 9.848e-07 \\ 
 5000 &  3 & 0.316 & 1.261 & 2.673e-08 & 1.389 & 6.447 & 3.612e-08 \\ 
 5000 &  6 & 0.544 & 1.231 & 2.323e-11 & 2.353 & 6.448 & 2.414e-10 \\ 
50000 &  2 & 2.337 & 123.757 & 6.179e-06 & 9.432 & 645.378 & 1.097e-05 \\ 
50000 &  3 & 3.556 & 124.390 & 2.132e-07 & 14.216 & 644.941 & 5.763e-07 \\ 
50000 &  6 & 7.844 & 127.218 & 4.156e-11 & 32.980 & 645.023 & 9.362e-10 \\ 
500000&  2 & 21.892 & 13045.236 & 3.397e-06 & 96.022 & 66505.855 & 3.703e-05\\ 
500000&  3 & 54.362 & 13499.781 & 9.809e-08 & 234.691 & 65616.089 & 1.589e-06\\ 
500000&  6 & 84.764 & 13593.953 & 6.365e-11 & 370.229 & 65679.543 & 1.649e-09\\ 

\hline 

\end{tabular}
\end{center}
\end{table}

\begin{table}
\caption{Timing results for elastostatic double-layer sources in a half-space (Mindlin's solution).
\label{tab4}}
\begin{center} 
\begin{tabular}{r|r|r|r|r}

\hline 

$N$ & $Prec$ & $T^{Mindlin}_{FMM}$ & $T^{Mindlin}_{dir}$ & $E_{Mindlin}$ \\

\hline 

 5000 &  2 & 1.067 & 72.212 & 8.097e-06 \\ 
 5000 &  3 & 2.222 & 71.808 & 2.975e-07 \\ 
 5000 &  6 & 5.250 & 71.918 & 6.551e-10 \\ 
50000 &  2 & 14.463 & 7181.319 & 3.788e-05 \\ 
50000 &  3 & 24.926 & 7182.631 & 3.737e-06 \\ 
50000 &  6 & 63.668 & 7171.081 & 1.340e-09 \\ 
500000 &  2 & 154.798 & 718416.170 & 7.103e-05 \\ 
500000 &  3 & 373.080 & 719714.520 & 1.962e-06 \\ 
500000 &  6 & 707.994 & 717146.230 & 1.327e-09 \\ 

\hline 

\end{tabular}
\end{center}
\end{table}

\section{Conclusions }\label{sec:conclusions}

In this paper, we have presented a fast multipole method for elastostatic 
interactions using Mindlin's solution --- the Green's function that satisfies
the condition of zero normal stress in a half-space.
We hope that the algorithm will prove useful in geophysical modeling.

\section{Acknowledgements}
We thank Shidong Jiang, Michael Minion, Michael Barall, Terry Tullis, Jim
Dieterich, and Keith Richards-Dinger for useful conversations.
This work was supported by the National Science Foundation under Grant
DMS-0934733 and by the Department of Energy under contract
DEFG0288ER25053.

\end{document}